\documentclass[twoside, 12pt]{article}
\usepackage{mathrsfs}
\usepackage{amssymb, amsmath, mathrsfs, amsthm}
\usepackage{graphicx}
\usepackage{epstopdf}
\usepackage{color}
\usepackage[top=2cm, bottom=2cm, left=2.5cm, right=2.5cm]{geometry}
\usepackage{float, caption, subcaption}
\usepackage{tikz}
\usepackage{cite}
\usepackage{leftidx}
\usepackage[section]{algorithm}
\usepackage{algorithmic}
\usepackage{multicol}
\usepackage{multirow}
\usepackage{makecell}
\usepackage{caption}
\newcommand{\bd}{\begin{description}}
\newcommand{\ed}{\end{description}}
\newcommand{\bi}{\begin{itemize}}
\newcommand{\ei}{\end{itemize}}
\newcommand{\be}{\begin{enumerate}}
\newcommand{\ee}{\end{enumerate}}
\newcommand{\beq}{\begin{equation}}
\newcommand{\eeq}{\end{equation}}
\newcommand{\beqs}{\begin{eqnarray*}}
\newcommand{\eeqs}{\end{eqnarray*}}

\definecolor{DarkGreen}{rgb}{0.2, 0.6, 0.3}

\newtheorem{theorem}{Theorem}[section]

\newtheorem{definition}{Definition}
\newtheorem{corollary}[theorem]{Corollary}

\newtheorem{example}{Example}

\newtheorem{claim}{Claim}

\newtheorem{fact}{Fact}
\newtheorem{proposition}{Proposition}[section]

\begin{document}
\title{\textbf{Monitoring the edges of product networks using distances}
\textbf{\footnote{Supported by the National Science Foundation of China (No. 12061059)
and Fundamental Research Funds for the Central Universities (No. 2682020CX60).}}}
\author{Wen Li\footnote{School of Computer, Qinghai Normal
University, Xining,Qinghai 810008, China. {\tt Email: liwzjn@163.com.}}
~\footnote{School of Mathematics and Statistics, Qinghai Minzu
University, Xining, Qinghai 810008, China.}\ ,
\ \ Ralf Klasing\footnote{Universit\'{e} de Bordeaux, Bordeaux INP, CNRS, LaBRI, UMR 5800, Talence, France.
{\tt Email: ralf.klasing@labri.fr.}}\ ,
\ \ Yaping Mao\footnote{Academy of Plateau
Science and Sustainability, Qinghai Normal University, Xining, Qinghai 810008, China.
{\tt yapingmao@outlook.com.}}\ ,
\ \ Bo Ning\footnote{Corresponding author. School of Computer Science, Nankai University, Tianjin 300350, China.
{\tt Email: bo.ning@nankai.edu.cn (B. Ning)}}}
\date{}
\maketitle{}

\begin{abstract}
Foucaud {\it et al.}
recently introduced and initiated the study of a new graph-theoretic concept in the area of network monitoring.
Let $G$ be a graph with vertex set $V(G)$, $M$ a subset of $V(G)$,
and $e$ be an edge in $E(G)$, and let $P(M, e)$
be the set of pairs $(x,y)$ such that $d_G(x, y)\neq  d_{G-e}(x, y)$
where $x\in M$ and $y\in V(G)$.
$M$ is called a \emph{distance-edge-monitoring set}
if every edge $e$ of $G$ is monitored by some vertex of $M$,
that is, the set $P(M, e)$ is nonempty. The
{\em distance-edge-monitoring number} of $G$, denoted by $\operatorname{dem}(G)$,
is defined as the smallest size of distance-edge-monitoring sets of $G$.
For two graphs $G,H$ of order $m,n$,
respectively, in this paper we prove that
$\max\{m\operatorname{dem}(H),n\operatorname{dem}(G)\}
\leq\operatorname{dem}(G\,\Box \,H)
\leq m\operatorname{dem}(H)+n\operatorname{dem}(G)
-\operatorname{dem}(G)\operatorname{dem}(H)$, where $\Box$
is the Cartesian product operation. Moreover, we characterize
the networks attaining the upper and lower bounds and show their
applications on some known networks. We also obtain
the distance-edge-monitoring numbers of join, corona, cluster,
and some specific networks.\\[2mm]
\end{abstract}

\section{Introduction}

In the area of network monitoring,
Foucaud {\it et al.}~\cite{FF21} introduced a new graph
theoretic concept, called \emph{distance-edge-monitoring}.
It aims to monitor some network using distance probes.
One wishes to monitor the network when a connection
(an edge) fails, that means this failure can be detected.
We hope to select a smallest set of vertices
of the network, called \emph{probes}, which will monitor all edges of the network.
At any given moment, a probe of the network can
measure its distance to any other vertices of the network.
The aim is that, whenever some edge of the network fails,
one of the measured distances changes, and then the probes
are able to detect the failure of any edge.

In real life networks, using probes to measure distances in a network is common.
For instance, this is useful in the fundamental
task of routing \cite{DABVV06,GT00}. It is also usually used for problems
concerning network verification \cite{BBDGKP15,BEEHHMR06,BEEHHMR10}.
In this paper, the networks are modeled by finite
undirected simple connected graph, with vertices
representing computers and edges representing connections
between computers.

\subsection{Distance-edge-monitoring set and number}

All graphs considered in this paper are undirected, finite, connected and
simple. We refer to the textbook \cite{JB08} for graph theoretical
notation and terminology not defined here. For a graph $G$, we use
$V(G)$ and $E(G)$ to denote the set of vertices and edges,
respectively. The number of vertices in $G$ is the \emph{order} of
$G$. For any set $X\subseteq V(G)$, let $G[X]$ denote the subgraph
induced by $X$; similarly, for any set $F\subseteq E(G)$, let $G[F]$
denote the subgraph induced by $F$. For an edge set $X\subseteq E(G)$,
let $G-X$ denote the subgraph of $G$ obtained by removing all the
edges in $X$ from $G$. If $X=\{x\}$, we may use $G-x$ instead of $G-\{x\}$.
For a vertex set $Y\subseteq V(G)$,
let $G\setminus Y$ denote the subgraph of $G$ obtained by removing all the
vertices and together with the edges
that incident to the vertex in $Y$ from $G$.
For two subsets $X$ and $Y$ of $V(G)$, we denote by $E_G[X,Y]$
the set of edges of $G$ with one end in $X$ and the other one in $Y$.
The set of neighbors of a vertex $v\in G$ is denoted by $N_G(v)$.
Moreover, $N_G[v]=N_G(v)\cup \{v\}$. A path of order $n$ is
denoted by $P_n$. The distance between vertices $u$ and $v$
of a graph $G$, denoted by $d_G(u, v)$, is the length of a shortest
$u,v$-path. Sometimes one can simply write $d(u, v)$
instead of $d_G(u, v)$. The \emph{eccentricity} $e(v)$ of
a vertex $v$ in a connected graph $G$ is the distance between
$v$ and a vertex farthest from $v$ in $G$. Then the \emph{radius}
of $G$, denoted by $rad(G)$, is the smallest eccentricity among the
vertices of $G$.

First, some definitions and useful results
about \emph{distance-edge-monitoring}
are introduced as follows.

\begin{definition}{\upshape\cite{FF21}}\label{pme}
For a set $M$ of vertices and an edge $e$ of a graph $G$,
let $P_G(M,e)$ be the set of pairs $(x, y)$ with $x$ a vertex
of $M$ and $y$ a vertex of $V(G)$ such that
$d_G(x,y)\neq d_{G-e}(x,y)$.
\end{definition}

Definition \ref{pme} means that $e$ belongs to
all shortest paths between $x$ and $y$ in $G$.

\begin{definition}{\upshape\cite{FF21}}
For a vertex $x$, let $EM_G(x)$ be the set of edges $e$ such
that there exists a vertex $v$ in $G$ with $(x,v)\in P({x},e)$.
\end{definition}

If $e\in EM_G(x)$, we say that \emph{$e$ is monitored by $x$}.
Sometimes one can simply write $P(M,e)$ and $EM(x)$
instead of $P_G(M,e)$ and $EM_G(x)$, respectively.

\begin{definition}{\upshape \cite{FF21}}\label{demset}
A set $M$ of vertices of a graph $G$ is called a \emph{distance-edge-monitoring
set} if any edge $e$ of $G$ is monitored by some vertex of $M$.
\end{definition}

For any edge $e\in E(G)$, Definition \ref{demset} means that  $P(M,e)\neq\emptyset$.

\begin{definition}{\upshape \cite{FF21}}
The distance-edge-monitoring number $\operatorname{dem}(G)$ of a graph
$G$ is defined as the smallest size of distance-edge-monitoring sets.
\end{definition}

The vertices of $M$ represent distance probes in a network modeled by $G$;
when the edge $e$ fails, the distance from $x$ to $y$ increase, and
thus we are able to detect the failure. It turns out that not only
we can detect it, but we can even correctly locate the failing
edge~\cite{FF21}. A detailed discussion of {\it distance-edge-monitoring sets}
and related concepts can be found in \cite{FF21}.

A set $C$ of vertices is a \emph{vertex cover}
of $G$ if every edge of $G$ has one of its endpoints in $C$.
The \emph{vertex cover number} of $G$, denoted by $c(G)$,
is the smallest size of a vertex cover of $G$.

\begin{theorem}{\upshape \cite{FF21}}\label{c}
Let $G$ be a graph of order $n$. Any vertex cover $C$ of $G$ is a
distance-edge-monitoring set, and hence $\operatorname{dem}(G)\leq c(G)\leq n-1$.
Moreover, $\operatorname{dem}(G)=n-1$ if and only if $G=K_n$.
\end{theorem}

\begin{theorem}{\upshape \cite{FF21}}\label{K_ab}
Let $K_{m,n}$ be the complete bipartite graph with parts
of sizes $m$ and $n$. Then $\operatorname{dem}(K_{m,n})=c(K_{m,n})=\min\{m, n\}$.
\end{theorem}

\subsection{Graph products}

Product networks were proposed based upon the idea
of using the cross product as a tool for ``combining''
two known graphs with established properties to obtain
a new one that inherits properties \cite{DAA97}.
In Graph Theory, Cartesian product is one of the main products,
and has been studied extensively \cite{IKR08,RYKO14,WMEZ19}.
For more details on graph products, we
refer to the monograph \cite{HIK11}.

In this paper, let $G$ and $H$ be two disjoint graphs with
$V(G)=\{u_1,\ldots,u_m\}$ and $V(H)=\{v_1,\ldots,v_n\}$.
Four binary operations of $G$ and $H$ considered
in this paper are defined as follows.

\begin{itemize}
\item[] The {\it join} or {\it complete product} of $G$ and $H$,
denoted by $G\vee H$, is the graph with vertex set $V(G)\cup V(H)$
and edge set $E(G)\cup E(H)\cup \{u_iv_j\,|\, u_i\in V(G), v_j\in V(H)\}$.

\item[] The {\it corona} $G*H$ is obtained by taking one copy of $G$
and $m$ copies of $H$, and by joining each vertex of the $i$-th
copy of $H$ with the $i$-th vertex of $G$, where
$i=1,2,\ldots,m$.

\item[] The {\it cluster} or {\it rooted product}
$G\odot H$ is obtained by taking one copy of $G$
and $m$ copies of a rooted graph $H$, and by identifying the
root of the $i$-th copy of $H$ with the $i$-th vertex of $G$, where
$i=1,2,\ldots,m$, and the root of all the copies of $H$ is the same.

\item[] The {\it Cartesian product} of $G$ and $H$,
written as $G\square H$, is the graph with vertex set
$V(G)\times V(H)=\{w_{i,j}\,|\,(u_i,v_j),u_i\in V(G),v_j\in V(H)\}$,
in which two vertices $w_{i,j}$ and $w_{i',j'}$ are adjacent if
and only if $u_i=u_{i'}$ and $v_jv_{j'}\in E(H)$,
or $v_j=v_{j'}$ and $u_iu_{i'}\in E(G)$.
\end{itemize}

Foucaud {\it et al.}~\cite{FF21} introduced and initiated the study of distance-edge-monitoring sets.
They showed that
for a nontrivial connected graph $G$ of order $n$, $1\leq \operatorname{dem}(G)\leq n-1$
with $\operatorname{dem}(G)=1$ if and only if $G$ is a tree,
and $\operatorname{dem}(G)=n-1$ if and only if it is a complete graph.
They derived the exact value of $\operatorname{dem}$ for grids,
hypercubes, and complete bipartite graphs.
Meantime, they related $\operatorname{dem}$ to other standard graph
parameters, and showed that $\operatorname{dem}(G)$ is lower-bounded
by the arboricity of the graph, and upper-bounded by its vertex
cover number. It is also upper-bounded by twice its feedback
edge set number. Moreover, they characterized connected graphs
$G$ with $\operatorname{dem}(G)=2$.
For the aspects of algorithm, they showed that determining
$\operatorname{dem}(G)$ for an input graph $G$ is an NP-complete problem,
even for apex graphs. There exists a polynomial-time logarithmic-factor
approximation algorithm, however it is NP-hard to compute an
asymptotically better approximation, even for bipartite graphs
of small diameter and for bipartite subcubic graphs. For such
instances, the problem is also unlikely to be fixed parameter
tractable when parameterized by the solution size.

In this paper, we study the \emph{distance-edge-monitoring numbers}
of join, corona, cluster, and Cartesian product,
and obtain the exact values of some specific networks.

\section{Join, corona, and cluster operations}
In this section, we obtain some
results about distance-edge-monitoring number
on some binary operations.
Firstly, some theorems of \cite{FF21}
given below are very helpful for our proof.

\begin{theorem}{\upshape \cite{FF21}}\label{EM}
Let $G$ be a connected graph and $x\in V(G)$.
The following two conditions are equivalent.
\begin{itemize}
\item[]$(1)$ $EM(x)$ is the set of edges incident to $x$.

\item[]$(2)$ For every vertex $y$ of $G$ with
$y\in V(G)-N_G[x]$, there exist two
shortest paths from $x$ to $y$ sharing at most one edge
(the one incident to $x$).
\end{itemize}
\end{theorem}

\begin{theorem}{\upshape \cite{FF21}}\label{dem=c}
If $EM(x)$ consists exactly of the edges incident
to any vertex $x$ of a graph $G$,
then a set $M$ is a distance-edge-monitoring
set of $G$ if and only if it is a vertex cover of $G$.
\end{theorem}

\begin{theorem}{\upshape \cite{FF21}}\label{pro1}
For any graph $G$, we have $c(G)\leq \operatorname{dem}(G\vee K_1)\leq c(G)+1$.
Moreover, if $rad(G)\geq 4$, then $\operatorname{dem}(G\vee K_1)=c(G)$.
\end{theorem}

Next we study the distance-edge-monitoring number
on the join operation of two graphs.

\begin{theorem}
Let $G$ and $H$ be two graphs with order $m,n\geq 2$, respectively. Then
$$
\operatorname{dem}(G\vee H)=c(G\vee H)=\min\{c(G)+n,c(H)+m\}.
$$
\end{theorem}
\begin{proof}
Since the distance between any two non-adjacent
vertices of $G\vee H$ is two, it follows that there are
at least two internally disjoint shortest paths between any two
non-adjacent vertices. By Theorems \ref{EM} and \ref{dem=c},
for any vertex $x\in V(G\vee H)$, $EM(x)$ is the set of edges
incident to $x$. That is, a minimum vertex cover of $G\vee H$
is a distance-edge-monitoring set of $G\vee H$.
Hence, we have $\operatorname{dem}(G\vee H)=c(G\vee H)$.

Without loss of generality, let $c(H)+m\leq c(G)+n$.
For any vertex $x\in V(G\vee H)$, we have
$EM(x)=\{xy\,|\,y\in N_{G\vee H}(x)\}$.
Since all the edges in
$\{uv\in E(G\vee H)\,|\,\ u\in V(G),v\in V(H)\}$
need all the $m$ monitoring vertices in $G$,
it follows that the remaining edges in $E(H)$
need at least $c(H)$ monitoring vertices, and hence
$\operatorname{dem}(G\vee H)\geq c(H)+m$.
To show $\operatorname{dem}(G\vee H)\leq c(H)+m$, we
suppose that $C$ is a vertex cover of $H$.
Let $M=V(G)\cup C$. Clearly, $M$ is a
distance-edge-monitoring set of $G\vee H$, and so
$\operatorname{dem}(G\vee H)\leq c(H)+m$,
which implies $\operatorname{dem}(G\vee H)=c(H)+m$.
The case $c(H)+m\geq c(G)+n$ can be discussed similarly.
\end{proof}

Then we study the distance-edge-monitoring number
on the corona of $G$ and $H$.
Let $V(G*H)=V(G)\cup\{w_{i,j}\,|\,1\leq i\leq m, \ 1\leq j\leq n\}$.
For $u_i\in V(G)$, we use $H_i$ to
denote the subgraph of $G*H$ induced by the vertex set
$\{w_{i,j}\,|\,1\leq j\leq n\}$.
For a fixed integer $i \ (1\leq i\leq m)$,
these edges $u_iw_{i,j}\in E(G*H)$ for each $j \ (1\leq j\leq n)$
are denoted by $E_i=\{u_iw_{i,j}\ |\ j=1,2,\ldots,n\}$.
Then $V(G*H)=V(G)\cup\left(\bigcup_{i=1}^{m} V(H_i)\right)$,
and $E(G*H)=E(G)\cup\left(\bigcup_{i=1}^{m} E(H_i)\right) \cup \left(\bigcup_{i=1}^{m}E_i\right)$.

\begin{theorem}
Let $G$ and $H$ be two graphs with order $m,n\geq 2$, respectively.
Then
$$\operatorname{dem}(G*H)=mc(H).$$
\end{theorem}
\begin{proof}
For any vertex $w_{s,t}\in H_s$,
where $1\leq s\leq m, 1\leq t\leq n$, we have the following fact.

\begin{fact}\label{fact1}
$\left(\bigcup_{i=1,i\neq s}^{m}E_i\right)\cup
E[\{u_s\},N_G(u_s)]
\subseteq EM_{G*H}(w_{s,t}).$
\end{fact}
\begin{proof}
First, we want to show that
$\bigcup_{i=1,i\neq s}^{m}E_i\subseteq EM_{G*H}(w_{s,t})$.
For any edge of $\bigcup_{i=1,i\neq s}^{m}E_i$,
say $u_iw_{i,j}$,
since $d_{G*H}(w_{s,t},w_{i,j})
=d_{G*H}(w_{s,t},u_i)+1$, where $i\neq s$,
it follows that
\begin{align*}
  & d_{G*H-u_iw_{i,j}}(w_{s,t},w_{i,j})
  \geq d_{G*H}(w_{s,t},u_i)+2
  >d_{G*H}(w_{s,t},u_i)+1=d_{G*H}(w_{s,t},w_{i,j}).
\end{align*}

Thus $u_iw_{i,j}\in EM_{G*H}(w_{s,t})$,
and so $\bigcup_{i=1,i\neq s}^{m}E_i\subseteq EM_{G*H}(w_{s,t})$.

Next we need to prove that
$E[\{u_s\},N_G(u_s)]\subseteq EM_{G*H}(w_{s,t})$.
For any edge $u_su_k\in E[\{u_s\},N_G(u_s)]$,
since $d_{G*H}(w_{s,t},u_k)=d_{G*H}(w_{s,t},u_s)+1$,
it follows that
$$
d_{G*H-u_su_k}(w_{s,t},u_k)\geq d_{G*H}(w_{s,t},u_s)+2>d_{G*H}(w_{s,t},u_s)+1=d_{G*H}(w_{s,t},u_k).
$$
Therefore $u_su_k\in EM_{G*H}(w_{s,t})$,
and so $E[\{u_s\},N_G(u_s)]\subseteq EM_{G*H}(w_{s,t})$.
Hence we have
$$
\left(\bigcup_{i=1,i\neq s}^{m}E_i\right)\cup
E[\{u_s\},N_G(u_s)]
\subseteq EM_{G*H}(w_{s,t}).
$$
\end{proof}

It is clear that for any $H_i$,
the edge $w_{i,j}w_{i,j'}$ cannot be monitored
by some vertex in $V(G*H)- V(H_i)$.
For two non-adjacent vertices $w_{i,j}$
and $w_{i,j'}\in V(H_i)$,
if $d_{H_i}(w_{i,j},w_{i,j'})=2$,
then there are two internally disjoint
shortest paths with length two;
if $d_{H_i}(w_{i,j},w_{i,j'})\geq 3$,
then there are only one
shortest paths $w_{i,j}u_iw_{i,j'}$ with length two.
So only $w_{i,j}$ or $w_{i,j'}$
can monitor the edge $w_{i,j}w_{i,j'}$.
This means that all edges in $H_i$ can be monitored by a
cover set $C_i$ of $H_i$.
Hence, together with Fact \ref{fact1} that
$\bigcup_{i=1}^{m}C_i$ is a
distance-edge-monitoring set of $G*H$, and so
$\operatorname{dem}(G*H)\leq mc(H)$.

Then we want to show that
$\operatorname{dem}(G*H)\geq mc(H)$.
Suppose that $M'$ is a distance-edge-monitoring set
of $G*H$ with $|M'|=mc(H)-1$.
There exists at least a copy of $H_k$
such that $|M'\cap V(H_k)|\leq c(H)-1$.
Then there exists at least one edge $e$ of $H_k$
cannot be monitored by $M'\cap V(H_k)$.
Since all the edges of $H_k$ cannot be monitored by
any vertex of $V(G*H)-V(H_k)$,
it follows that $e$ cannot be monitored by $M'$, a contradiction.
Hence $\operatorname{dem}(G*H)\geq mc(H)$.
\end{proof}

Finally we study the distance-edge-monitoring number
on the cluster of two graphs $G$ and $H$.
Recall the definition of cluster, we have $V(G\odot
H)=\{w_{i,j}\,|\,1\leq i\leq m, \ 1\leq j\leq n\}$.
For $u_i\in V(G)$, we let
$H_i$ denote the subgraph of $G\odot H$ induced by the vertex
set $\{w_{i,j}\,|\,1\leq j\leq n\}$.
Without loss of generality, we
assume $w_{i,1}$ is the root of $H_i$ for each $u_i\in V(G)$.
Let $G_1$ be the graph induced by the vertices in
$\{w_{i,1}\,|\,1\leq i\leq m\}$. Clearly, $G_1\simeq G$, and
$V(G\odot H)=V(H_1)\cup V(H_2)\cup \ldots \cup V(H_n)$.

Foucaud et al. \cite{FF21} studied distance-edge-monitoring number
of trees and proved the following.

\begin{theorem}{\upshape \cite{FF21}}\label{tree}
Let $G$ be a connected graph with at least one edge.
Then $\operatorname{dem}(G)=1$ if and only if $G$ is a tree.
\end{theorem}

Let $G$ be a graph. The \emph{base graph} $G_b$ of $G$
is the graph obtained from $G$ by iteratively deleting
pendant edges together with vertices of degree $1$.
Then Foucaud et al. \cite{FF21}
derived the following.

\begin{theorem}{\upshape \cite{FF21}}\label{BG}
Let $G_b$ be the base graph of a graph $G$.
Then we have $dem(G)=dem(G_b)$.
\end{theorem}

A tree, denoted by $T$, is a connected acyclic graph.
From Theorem \ref{TPmPn}, we can directly obtain
the following.
In this paper, the distance-edge-monitoring number of cluster
is determined below.

\begin{theorem}\label{cluster}
Let $G$ and $H$ be two graphs with order $m,n\geq 2$, respectively.
Then
$$
\operatorname{dem}(G)\leq \operatorname{dem}(G\odot H)
\leq m\operatorname{dem}(H).
$$
Moreover, $\operatorname{dem}(G\odot H)=\operatorname{dem}(G)$
if and only if $H$ is a tree.
\end{theorem}
\begin{proof}
Let $\operatorname{dem}(G)=d_1$ and
$\operatorname{dem}(H)=d_2$.
First, if $H$ is a tree, it follows
from Theorem \ref{BG} that $\operatorname{dem}(G\odot H)=d_1$.
Conversely, we want to show that if $\operatorname{dem}(G\odot H)=d_1$,
then $H$ is a tree. Assume that $H$ is not a tree, it follows from
Theorem \ref{tree} that $d_2\geq 2$. Suppose that $M'$ is a
distance-edge-monitoring set of $G\odot H$ with $|M'|=d_1$.
If $M'\subset V(G)$, then there is an edge $e$ of some copy of $H$
which cannot be monitored by $M'$ as $d_2\geq 2$, a contradiction.
If $|M'\cap V(G)|=d_1-a(1\leq a\leq d_1\leq m-1)$, then
$M''=M'- V(G)=\{w_{i,j}\,|\, 1\leq i\leq m, 2\leq j\leq n\}$,
and $|M''|=a$. Since the $a$ vertices in $M''$ are only allowed in
the $m$ copies of $H$, it follows that there exists a copy,
say $H_t$, such that $M''\cap V(H_t)=\emptyset$.
It is clearly that $EM_{H_t}(w_{i,j})=EM_{H_t}(w_{t,1})$,
where $i\neq t$ and $1\leq j\leq m$. Since $\operatorname{dem}(H_t)\geq 2$,
it follows that there is an edge $e\in E(H_t)$ which cannot be
monitored by $M'$, a contradiction. Hence $H$ is a tree.

Next, we want to prove that
$d_1+1\leq \operatorname{dem}(G\odot H)\leq md_2$,
where $H$ is not a tree.
For any vertex $w_{s,t}\in H_s$,
where $1\leq s\leq m, 1\leq t\leq n$,
we have the following fact.
\setcounter{fact}{0}
\begin{fact}\label{fact2}
$E[\{w_{1,t}\},N_G(w_{1,t})]
\subseteq EM_{G\odot H}(w_{s,t}).$
\end{fact}
\begin{proof}
If $s=1$, it is clear that the statement holds.
If $s\geq 2$, then for any edge
$w_{1,t}w_{1,k}\in E[\{w_{1,t}\},N_G(w_{1,t})]$,
since $d_{G\odot H}(w_{s,t},w_{1,k})=d_{G\odot H}(w_{s,t},w_{1,t})+1$,
it follows that
$$
d_{G\odot H-w_{1,t}w_{1,k}}(w_{s,t},w_{1,k})
\geq d_{G\odot H}(w_{s,t},w_{1,t})+2
>d_{G\odot H}(w_{s,t},w_{1,t})+1
=d_{G\odot H}(w_{s,t},w_{1,k}).
$$
Therefore $w_{1,t}w_{1,k}\in EM_{G\odot H}(w_{s,t})$,
and so $E[\{w_{1,t}\},N_G(w_{1,t})]
\subseteq EM_{G\odot H}(w_{s,t})$.
\end{proof}

Let $M_H$ be a distance-edge-monitoring set of $H$
with $|M_H|=d_2\geq 2$.
Then
$M_{H_i}$ is the distance-edge-monitoring set
of $H_i$ corresponding to $M_{H}$ in $H$.
where $1\leq i\leq m$.
For any two vertices $w_{i,j}$ and $w_{i,j'}$,
it follows from Fact \ref{fact2} and
$d_{G\odot H}(w_{i,j},w_{i,j'})=d_{H}(w_{i,j},w_{i,j'})$,
that the vertices in $M_{H_i}$
not only monitors the edges in $M_{H_i}$,
but also monitors the edges in
$E[\{w_{1,i}\},N_G(w_{1,i})]$.
Hence $M=\bigcup_{i=1}^m M_{H_i}$
is a distance-edge-monitoring set of $G\odot H$,
and so $\operatorname{dem}(G\odot H)\leq md_2$.

In the end, we want to prove that
$\operatorname{dem}(G\odot H)\geq d_1+1$.
For any vertex subset $M'\subset V(G\odot H)$
such that $|M'|=d_1$, we have the fact that $M'$ is not a
distance-edge-monitoring set of $G\odot H$.
If $M'\subset V(G)$, then there is an edge
$e$ of some copy of $H$ which cannot be monitored by $M'$.
Therefore, $M'$ is not a distance-edge-monitoring set of $G\odot H$.
If $|M'\cap V(G)|=d_1-a~(1\leq a\leq d_1\leq m-1)$,
then $$M''=M'\setminus V(G)=\{w_{i,j}\,|\, 1\leq i\leq m, 2\leq j\leq n\},$$
and $|M''|=a$. Since the $a$ vertices in $M''$ are only allowed in
the $m$ copies of $H$, it follows that there exists a copy,
say $H_t$, such that $M''\cap V(H_t)=\emptyset$.
It is clearly that $EM_{H_t}(w_{t,j})=EM_{H_t}(w_{t,1})$,
where $i\neq t$ and $1\leq j\leq m$.
Since $\operatorname{dem}(H_t)\geq 2$, it follows that
there is an edge $e\in E(H_t)$ which cannot be monitored by $M'$.
Therefore, $M'$ is not a distance-edge-monitoring set of $G\odot H$.
We have proved $\operatorname{dem}(G\odot H)\geq d_1+1$.
\end{proof}

\section{Results for Cartesian product}

In this section, we study $P(M,e)$, $EM(x)$
and distance-edge-monitoring number
of Cartesian product of two general graphs.
Furthermore, we determine distance-edge-monitoring numbers
of some specific networks.

For any vertex $v_j\in V(H)$, the subgraph of $G\Box H$
induced by the vertex set $\{w_{i,j}|1\leq i\leq m\}$ is
denoted by $G_j$. Similarly, for any vertex $u_i\in V(G)$,
the subgraph of $G\Box H$ induced by the vertex set
$\{w_{i,j}| 1\leq j\leq n\}$ is denoted by $H_i$.

\subsection{General results}

In this section, we firstly study $P(M,e)$ and $EM(x)$
of the Cartesian product of two general graphs.
We also determine distance-edge-monitoring number
of the Cartesian product of two general graphs.
Finally, we characterize the graphs
attaining the upper and lower bounds.

The following result is from the book \cite{IKR08}.

\begin{theorem}{\upshape \cite{IKR08}} \label{dis}
Let $w_{i,j}$ and $w_{i',j'}$ are two vertices of $G\Box H$.
Then
$$
d_{G\Box H}(w_{i,j},w_{i',j'})=d_G(u_i,u_{i'})+ d_H(v_j,v_{j'}).
$$
\end{theorem}

From the definition of the Cartesian product,
an edge of $G\Box H$ is either in the form
$w_{i,j}w_{i,j'}$ or $w_{i',j}w_{i,j}$. So we have the
following.

\begin{theorem}\label{PME}
Let $M$ be a vertex subset of $G\Box H$.
For any two edges $w_{i,j}w_{i,j'}$ and
$w_{i,j}w_{i',j}$ of $G\Box H$, we have that

$(i)$ $P_{G\Box H}(M,w_{i,j}w_{i,j'})=
P_{H_i}(M\cap V(H_i),w_{i,j}w_{i,j'})$;

$(ii)$ $P_{G\Box H}(M,w_{i,j}w_{i',j})=
P_{G_j}(M\cap V(G_j),w_{i,j}w_{i',j})$.
\end{theorem}
\begin{proof}
For some vertex $w_{a,b}\in M$ and $a\neq i$,
we claim that there is no vertex $w_{x,y}$ of $G\Box H$
such that $d_{G\Box H}(w_{a,b},w_{x,y})\neq
d_{G\Box H-w_{i,j}w_{i,j'}}(w_{a,b},w_{x,y})$.
If there is a vertex $w_{x,y}$ of $G\Box H$
such that $d_{G\Box H}(w_{a,b},w_{x,y})\neq
d_{G\Box H-w_{i,j}w_{i,j'}}(w_{a,b},w_{x,y})$,
in other words, $w_{i,j}w_{i,j'}$ belongs to
all shortest paths between $w_{a,b}$ and
$w_{x,y}$ in $G\Box H$. This is a contradiction
as we can find two shortest paths, say
$$
P_1=w_{a,b}\cdots w_{i,b}\cdots w_{i,j}w_{i,j'}\cdots w_{x,y},
$$
and
$$
P_2=w_{a,b}\cdots w_{a,j}w_{a,j'}\cdots w_{a,y}\cdots w_{x,y},
$$
such that the edge $w_{i,j}w_{i,j'}$ is not in one of them.
From Theorem \ref{dis}, we can see $P_1$ and $P_2$
are the shortest paths from $w_{a,b}$ to $w_{x,y}$.
But $w_{i,j}w_{i,j'}\in E(P_1)$ and $w_{i,j}w_{i,j'}\notin E(P_2)$.
Therefore, we only need to consider the vertices of
$M\cap V(H_i)$, that is
$P_{G\Box H}(M,w_{i,j}w_{i,j'})=
P_{G\Box H}(M\cap V(H_i),w_{i,j}w_{i,j'})$.

Next we want to show $P_{G\Box H}(M\cap H_i,w_{i,j}w_{i,j'})=
P_{H_i}(M\cap V(H_i),w_{i,j}w_{i,j'})$.
According to the above proof, we have the fact that
for any vertex $w_{a,b}=w_{i,b}\in M\cap V(H_i)$,
if there is a vertex $w_{x,y}$ of $G\Box H$ such that
$d_{G\Box H}(w_{a,b},w_{x,y})\neq
d_{G\Box H-w_{i,j}w_{i,j'}}(w_{a,b},w_{x,y})$,
then $w_{x,y}=w_{i,y}\in V(H_i)$,
that is
$$
P_{G\Box H}(M,w_{i,j}w_{i,j'})=
P_{G\Box H}(M\cap V(H_i),w_{i,j}w_{i,j'})=
P_{H_i}(M\cap V(H_i),w_{i,j}w_{i,j'}).
$$

By the symmetry, we can also
prove that
$$
P_{G\Box H}(M,w_{i,j}w_{i',j})=
P_{G_j}(M\cap V(G_j),w_{i,j}w_{i',j}).
$$
This proves the theorem.
\end{proof}

\begin{example}
Let $P_m$ and $P_n$ are two paths.
For any vertex set $M\subset V(P_m\Box P_n)$,
we have $P_{P_m\Box P_n}(M,w_{i,j}w_{i,j'})=
P_{P_n(u)}(M\cap V(P_n(u)),w_{i,j}w_{i,j'})$
and $P_{P_m\Box P_n}(M,w_{i,j}w_{i',j})=
P_{P_m(v)}(M\cap V(P_m(v)),w_{i,j}w_{i',j})$.
\end{example}

Note that if every vertex $w_{a,b}\in M$ with $a\neq i$,
then $P_{G\Box H}(M,w_{i,j}w_{i,j'})=\emptyset$.
If every vertex $w_{a,b}\in M$ with $b\neq j$,
then $P_{G\Box H}(M,w_{i,j}w_{i',j})=\emptyset$.

From Theorem \ref{PME}, we have
the following immediately.

\begin{corollary}\label{EMX}
For any vertex $w_{i,j}$ of $G\Box H$, we have
$$
EM_{G\Box H}(w_{i,j})=EM_{H_i}(u_i)\cup EM_{G_j}(v_j).
$$
\end{corollary}

\begin{theorem}\label{GH}
Let $G$ and $H$ be two graphs of order $m$ and $n$,
respectively. Then
$$
\max\{m\operatorname{dem}(H),n\operatorname{dem}(G)\}
\leq\operatorname{dem}(G\Box H)
\leq m\operatorname{dem}(H)+n\operatorname{dem}(G)
-\operatorname{dem}(G)\operatorname{dem}(H).
$$
Moreover, the bounds are sharp.
\end{theorem}
\begin{proof}
Let $\operatorname{dem}(G)=d_1$ and $\operatorname{dem}(H)=d_2$. Then we have the following claim.
\begin{claim}\label{claim1}
$\operatorname{dem}(G\Box H)\geq \max\{md_2,nd_1\}$.
\end{claim}
\begin{proof}
Assume, to the contrary, that $\operatorname{dem}(G\Box H)
\leq \max\{md_2,nd_1\}-1$.
If $md_2\leq nd_1$,
then $\operatorname{dem}(G\Box H)\leq nd_1-1$.
Suppose that $M$ is a distance-edge-monitoring set
with size $nd_1-1$ in $G\Box H$.
Then there exists some $G_j$,
such that $|V(G_j)\cap M|\leq d_1-1$.
Since $\operatorname{dem}(G)=\operatorname{dem}(G_j)=d_1$,
it follows that there exists some edge
$e_j\in E(G_j)$ such that $e_j$
cannot be monitored by $V(G_j)\cap M$.
Additionally, for any vertex $w_{i,s}$ of $G\Box H$,
where $s\neq j$, from Corollary \ref{EMX}, we have
$e_j\notin EM_{G\Box H}(w_{i,s})=EM_{H_i}(u_i)\cup EM_{G_s}(v_s)$.
Hence, the edge $e_j$ cannot be monitored by $M$ in $G\Box H$, a contradiction.
The case $md_2\geq nd_1$ also can be proved in the same way.
This proves Claim 1.
\end{proof}

From Claim \ref{claim1}, we have $\operatorname{dem}(G\Box H)\geq \max\{md_2,nd_1\}$.

For the upper bound,
it suffices to show that
$\operatorname{dem}(G\Box H)\leq
md_2+nd_1-d_1d_2$.
Without loss of generality, let $M_H^i=\{w_{i,1},w_{i,2},
\ldots,w_{i,d_2}\}$ be a monitor set of $H_i$,
where $1\leq i\leq m$.
Let $M_G^j=\{w_{1,j},w_{2,j},\ldots,w_{d_1,j}\}$
be a monitor set of $G_j$, where
$1\leq j\leq n$.
From Corollary \ref{EMX},
$\bigcup_{i=1}^m M_H^i$ can monitor all edges
in $\bigcup_{i=1}^m E(H_i)$ and
$\bigcup_{j=1}^{d_2} E(G_j)$, and
$\bigcup_{j={d_2}+1}^n M_G^j$ can monitor
all edges in $\bigcup_{j={d_2}+1}^n E(G_j)$.
Hence, $M=\big(\bigcup_{i=1}^m M_H^i\big)\cup
\big(\bigcup_{j={d_2}+1}^n M_G^j\big)$
is a distance-edge-monitoring set of
$G\Box H$ with size
$md_2+(n-d_2)d_1=md_2+nd_1-d_1d_2$, and so $\operatorname{dem}(G\Box H)\leq
md_2+nd_1-d_1d_2$.
\end{proof}

Next, the necessary and sufficient condition
for the upper bound of Theorem~\ref{GH} is given below.

\begin{theorem}\label{upper}
Let $G$ and $H$ be two graphs of order
$m$ and $n$, respectively. Then
$$\operatorname{dem}(G\Box H)=m\operatorname{dem}(H)
+n\operatorname{dem}(G)-\operatorname{dem}(G)\operatorname{dem}(H)$$
if and only if there is only one distance-edge-monitoring set in $G$ or $H$.
\end{theorem}
\begin{proof}
Let $\operatorname{dem}(G)=d_1$ and $\operatorname{dem}(H)=d_2$.
We first show that if
$\operatorname{dem}(G\Box H)=md_2+nd_1-d_1d_2$,
then there is only one distance-edge-monitoring set in $G$ or $H$.
Assume that $M_{G}^1$ and $M_{G}^2$
are two distance-edge-monitoring sets of $G$, and $M_{H}^1$ and $M_{H}^2$
are two distance-edge-monitoring sets of $H$, respectively. Let
$M_{G_j}^1=\{w_{s,j}\,|\,1\leq s\leq d_1\}$
and $M_{H_i}^1=\{w_{i,t}\,|\,1\leq t\leq d_2\}$,
where $1\leq i\leq m$, $1\leq j\leq n$, be the
distance-edge-monitoring sets of $G_j$ and $H_i$
corresponding to $M_{G}^1$ and $M_{H}^1$ in $G$ and $H$,
respectively. Let $x=|M_{G_j}^1- M_{G_j}^2|$ and
$y=|M_{H_i}^1- M_{H_i}^2|$.
Then $1\leq x\leq d_1$ and $1\leq y\leq d_2$.

From Corollary \ref{EMX},
$\bigcup_{j=1}^{d_2} M_{G_j}^1$ can monitor
all the edges in
$$
\{E(G_j)\,|\,1\leq j\leq d_2\}\cup \{E(H_i)\,|\,1\leq i\leq d_1\},
$$
and
$\bigcup_{j=d_2+1}^n M_{G_j}^2$
can monitor all the edges in
$\{E(G_j)\,|\,d_2+1\leq j\leq n\}$.

If $M_{H_i}^1\cap M_{H_i}^2=\emptyset$, then all the edges in
$\{E(H_i)\,|\,d_1+1\leq i\leq d_1+x\}$ can be monitored by
$\bigcup_{j=d_2+1}^n M_{G_j}^2$, and all the edges in
$\{E(H_i)\,|\,d_1+x+1\leq i\leq m\}$ can be monitored by
$\bigcup_{i=d_1+x+1}^m M_{H_i}^2$.
This means that
$$
M=\left(\bigcup_{j=1}^{d_2} M_{G_j}^1\right)
\cup \left(\bigcup_{j=d_2+1}^n M_{G_j}^2\right)
\cup \left(\bigcup_{i=d_1+x+1}^m M_{H_i}^2\right)
$$
is a distance-edge-monitoring set of $G\Box H$.
However, we have
$$|M|=nd_1+(m-d_1-x)d_2=nd_1+md_2-d_1d_2-xd_2<nd_1+md_2-d_1d_2,$$
a contradiction.

If $M_{H_i}^1\cap M_{H_i}^2\neq\emptyset$,
then all edges in $\{E(H_i)\,|\,d_1+1\leq i\leq m\}$
can be monitored by $\bigcup_{i=d_1+1}^m M_{H_i}^2$, and hence
$$
M=\left(\bigcup_{j=1}^{d_2} M_{G_j}^1\right)
\cup\left(\bigcup_{j=d_2+1}^n M_{G_j}^2\right)
\cup\left(\bigcup_{i=d_1}^m M_{H_i}^2\right)
$$
is a distance-edge-monitoring set of $G\Box H$. Since
$\left|\big(\bigcup_{j=d_2+1}^n M_{G_j}^2\big)
\cap\big(\bigcup_{i=d_1}^m M_{H_i}^2\big)\right|=xy$,
it follows that
$$|M|=nd_1+(m-d_1)d_2-xy=nd_1+md_2-d_1d_2-xy\leq nd_1+md_2-d_1d_2-1<nd_1+md_2-d_1d_2,$$
a contradiction.

By the above arguments, there is only one distance-edge-monitoring set in $G$ or $H$.

Conversely, we suppose that there is only one
distance-edge-monitoring set in $G$ or $H$. Without loss of generality, we suppose that there is only one
distance-edge-monitoring set in $G$, say $M_G=\{u_1,u_2,\ldots,u_{d_1}\}$. For each $j \ (1\leq j\leq n)$,
$M_{G_j}=\{w_{1,j},w_{2,j},\ldots,w_{d_1,j}\}$ is the
distance-edge-monitoring set of $G_j$
corresponding to $M_{G}$ in $G$.
\setcounter{claim}{0}
\begin{claim}\label{claim-cha}
$\operatorname{dem}(G\Box H)\geq md_2+nd_1-d_1d_2$.
\end{claim}
\begin{proof}
Assume, to the contrary, that $\operatorname{dem}(G\Box H)
\leq md_2+nd_1-d_1d_2-1$.
Let $M$ be a distance-edge-monitoring set
of size $md_2+nd_1-d_1d_2-1$ of $G\Box H$.
Let $M_{H}$ be a
distance-edge-monitoring set of $H$. For each $i \ (1\leq i\leq m)$,
$M_{H_i}=\{w_{i,1},w_{i,2},\ldots,w_{i,d_2}\}$ is the
distance-edge-monitoring set of $H_i$ corresponding to $M_{H}$ in $H$.
From Corollary \ref{EMX},
each $M_{G_j}$ can only monitor all the edges in $E(G_j)$,
and hence $\bigcup_{j=1}^n M_{G_j}$ can monitor all the edges in
$$
\{E(G_j)\,|\,1\leq j\leq n\}\cup \{E(H_i)\,|\,1\leq i\leq d_1\},
$$
and so the vertex set $M'=M-\bigcup_{j=1}^n M_{G_j}$ with size
$$(md_2+nd_1-d_1d_2-1)-nd_1=md_2-d_1d_2-1=(m-d_1)d_2-1$$
monitors all the edges in
$\{E(H_i)\,|\,d_1+1\leq i\leq m\}$.
Hence, there exists some $H_i$
such that $|H_i\cap M'|<d_2-1$.
In other words, there is some edge
$e\in E(H_i)$ which cannot be monitored
by $M$ in $G\Box H$, a contradiction.
\end{proof}

From Claim \ref{claim-cha}, we have
$\operatorname{dem}(G\Box H)
\geq md_2+nd_1-d_1d_2$.
From Theorem \ref{GH}, we have
$\operatorname{dem}(G\Box H)
\leq md_2+nd_1-d_1d_2$, and hence
$\operatorname{dem}(G\Box H)
=md_2+nd_1-d_1d_2$.
\end{proof}

To show the upper bound of Theorem \ref{GH}
is sharp, we recall the \emph{book graph} initially introduced
by Erd\H{o}s\cite{P62} in 1962.
A book graph $B_q$ consists of
$q$ triangles sharing a common edge.
In the following, the distance-edge-monitoring set
of $B_n$ is determined.

\begin{proposition}\label{B_n}
For a book graph $B_n$ with $n$ pages, we have
$$
\operatorname{dem}(B_n)=2.
$$
Moreover, the two vertices of the common edge in $B_n$
form the unique distance-edge-monitoring set
of $B_n$.
\end{proposition}
\begin{proof}
Let $uv$ be the common edge of $B_n$,
and let $V(B_n)- \{u,v\}=\{u_i\,|\,1\leq i\leq n\}$.
Each triangle of $B_n$
induced by the vertex set $\{u,v,u_i\}$ is
denoted by $K_3^i$, where $1\leq i\leq n$.
For each $K_3^i$, we need at least
two vertices to monitor all the three edges,
and so $\operatorname{dem}(B_n)\geq 2$.
Let $M=\{u,v\}$. Then $M$ can monitor all the edges
of $B_n$, and hence $\operatorname{dem}(B_n)\leq 2$, and so $\operatorname{dem}(B_n)=2$.

We have already proved that $M=\{u,v\}$ is a
distance-edge-monitoring set of $B_n$.
It suffices to show that $M$ is a unique
distance-edge-monitoring set of $B_n$.
Suppose that $M'$ is another
distance-edge-monitoring set of $B_n$.
Then we have $0\leq |M\cap M'|\leq 2$.
Suppose that $|M\cap M'|=0$. Without loss of generality,
let $M'=\{u_a,u_b\}$,
where $1\leq a,b\leq n$ and $a\neq b$.
Since $|M'\cap K_3^a|=|M'\cap K_3^b|=2$,
it follows that $M'$ can monitor all
the edges of $K_3^a$ and $K_3^b$.
However, $|M'\cap K_3^i|=0$, and so
$M'$ cannot monitor all
the edges of $K_3^i$, where $i=\{1,2,\ldots,n\}-\{a,b\}$,
a contradiction.
Suppose that $|M\cap M'|=1$. Without loss of generality,
let $M'=\{u,u_a\}$, where $1\leq a\leq n$.
Clearly, since $|M'\cap K_3^a|=2$, it follows that  $M'$ can monitor all
the edges of $K_3^a$. However, $|M'\cap K_3^j|=1$, and
so $M'$ cannot monitor all
the edges of $K_3^j$, where $j=\{1,2,\ldots,n\}-a$,
a contradiction.
If $|M\cap M'|=2$ then $M=M'$, and hence $M=\{u,v\}$ is the unique
distance-edge-monitoring set of $B_n$.
\end{proof}

From Theorem \ref{upper} and Proposition \ref{B_n},
the following results can show that
the upper bound of Theorem \ref{GH} is sharp.

\begin{corollary}
For two book graphs $B_n$ and $B_m$, we have
$$
\operatorname{dem}(B_n\Box B_m)=2m+2n+4.
$$
\end{corollary}

\begin{corollary}
For a connected graph $G$ with order $m$
and a book graph $B_n$, we have
$$
\operatorname{dem}(G\Box B_n)=
2m+(n+2)\operatorname{dem}(G)
-2\operatorname{dem}(G).
$$
\end{corollary}

Then the necessary and sufficient condition
for the lower bound of Theorem~\ref{GH} is given below.

\begin{theorem}\label{lower}
Let $G$ and $H$ be two graphs of order $m$ and $n\ (m\leq n)$,
respectively, such that $\operatorname{dem}(G)\geq \operatorname{dem}(H)$.
Then $\operatorname{dem}(G\Box H)=n\operatorname{dem}(G)$
if and only if all of the following hold.

$(1)$ For any vertex $x\in V(G)$,
there exists some distance-edge-monitoring set $M_G^i$
with $|M_G^i|=d_1$
such that $x\in M_G^i$.

$(2)$ $H$ has at least $k$ distance-edge-monitoring sets,
say $M_H^j\ (j=1,2,\ldots,k)$, such that
$M_H^p\cap M_H^q=\emptyset\ (1\leq p,q\leq k,\ p\neq q)$
for any two sets $M_H^p$ and $M_H^q$,
where $k=\min\{s\,|\,\bigcup_{i=1}^s M_G^i=V(G)\}\geq 2$
and $|M_H^j|=d_2$.
\end{theorem}
\begin{proof}
Let $\operatorname{dem}(G)=d_1$ and
$\operatorname{dem}(H)=d_2$.
Suppose that $\operatorname{dem}(G\Box H)=nd_1$.
Let $M$ be the distance-edge-monitoring set of $G\Box H$,
with $|M|=nd_1$.
It is clear that each $G_j$ has $d_1$ vertices to
monitor all edges in $G_j$, where $1\leq j\leq n$.
\setcounter{claim}{0}
\begin{claim}\label{claim-lower-1}
For any vertex $x\in V(G)$,
there exists some distance-edge-monitoring set $M_G^i$
with $|M_G^i|=d_1$
such that $x\in M_G^i$.
\end{claim}
\begin{proof}
Assume, to the contrary,
that there exists a vertex
$u_a\in V(G)$ such that $u_a$ is not in any
distance-edge-monitoring set of $G$.
Without loss of generality,
let $w_{a,j}$ be the vertex of $G_j$
corresponding to $u_a$ in $G$.
Then all vertices of $M$ are not in $H_a$.
From Corollary \ref{EMX},
$M$ cannot monitor the edges of $H_a$, a contradiction.
\end{proof}

By Claim \ref{claim-lower-1}, $(1)$ holds. By $(1)$,
there exist $s$ distance-edge-monitoring sets,
say $M_G^1,M_G^2,...,M_G^s$,
such that $\bigcup_{i=1}^s M_G^i=V(G)$, where $s\geq 2$.
Let $k=\min\{s\,|\,\bigcup_{i=1}^s M_G^i=V(G)\}$.
So $k\geq 2$.

\begin{claim}\label{claim-lower-2}
$H$ has at least $k$ distance-edge-monitoring sets
and each set has $d_2$ vertices.
\end{claim}
\begin{proof}
Assume, to the contrary, that $H$ has at most $k-1$
distance-edge-monitoring sets,
say $M_H^j\ (\ j=1,2,\ldots,k-1)$.
Let $M_H^j=\{v_j^1,v_j^2,\ldots,v_j^{d_2}\}$,
$M_G^i=\{u_i^1,u_i^2,\ldots,u_i^{d_1}\}$, and
$M^i=M_G^i\times M_H^i$.
Let $M_{H_i}^j\ (\ j=1,2,\ldots,k-1)$ be the
distance-edge-monitoring sets of $H_i$
corresponding to $M_{H}^j$ in $H$.

If $\bigcup_{t=1}^{k-1}M_H^t=V(H)$,
then $M'=\bigcup_{i=1}^{k-1} M^i$
can monitor at most all the edges in
$$
\{E(G_i), E(H_j)\,|\,1\leq i\leq n, 1\leq j\leq l-1\},
$$
where $|V(G)-\bigcup_{i=1}^{k-1} M_G^i|=l$.
Since $k=\min\{s\,|\,\bigcup_{i=1}^s M_G^i=V(G)\}\geq 2$,
it follows that the edges in
$\{E(H_j)\,|\, l\leq j\leq m\}$
cannot be monitored by $M'$.
So
$\operatorname{dem}(G\Box H)\geq nd_1+1$, a contradiction.

If $|\bigcup_{t=1}^{k-1} M_H^t- V(H)|=r\geq 1$,
then $M'=(\bigcup_{i=1}^{k-1} M^i)\cup M_G^k$
can monitor at most all the edges in
$\{E(G_i), E(H_j)\,|\,1\leq i\leq n, 1\leq j\leq l-1\}$,
where $|V(G)-\bigcup_{i=1}^{k-1} M_G^i|=l$.
Since the edges in
$\{E(H_j)\,|\, l\leq j\leq m\}$
cannot be monitored by $M'$, it follows that
$\operatorname{dem}(G\Box H)\geq nd_1+1$, a contradiction.
\end{proof}

From Claim \ref{claim-lower-2}, let $M_{H}^j \ (j=1,2,\ldots,k)$ be the
distance-edge-monitoring sets of $H$.
\begin{claim}\label{claim-lower-3}
$M_H^p\cap M_H^q=\emptyset$ for each $1\leq p\neq q\leq k$.
\end{claim}
\begin{proof}
Assume, to the contrary, that
there exist two distance-edge-monitoring sets
$M_H^p,M_H^q$ such that $M_H^p\cap M_H^q\neq \emptyset$.
Without loss of generality,
let $|M_H^1\cap M_H^2|=f(1\leq f\leq d_2-1)$,
and let $M_H^1=\{v_1,v_2,\ldots,v_{d_2}\}$,
$M_H^2=\{v_{d_2-f},v_{d_2-f+1},\ldots,v_{d_2},v_{d_2+1},\ldots,v_{2d_2-f}\}$.
Similarly, let $|M_G^1\cap M_G^2|=l(0\leq l\leq d_1-1)$,
and let $M_G^1=\{u_1,u_2,\ldots,u_{d_1}\}$,
$M_G^2=\{u_{d_1-l},u_{d_1-l+1},\ldots,u_{d_1},u_{d_1+1},\ldots,u_{2d_1-l}\}$.
Then $M_{G_i}^1$ and $M_{G_i}^2$ are
the distance-edge-monitoring sets of $G_i$
corresponding to $M_{G}^1$ and $M_{G}^2$ in $G$.
Then $M_{H_i}^1$ and $M_{H_i}^2$ are
the distance-edge-monitoring sets of $H_i$
corresponding to $M_{H}^1$ and $M_{H}^2$ in $H$,
respectively.
For any distance-edge-monitoring
set $M$ with size $nd_1$ of $G\Box H$,
since $\operatorname{dem}(G\Box H)=nd_1$,
it follows from Corollary \ref{EMX} that,
$|M\cap G_j|=d_1$, where $j=1,2,\ldots,n$.
For all the $H_i, i=1,2,\ldots,2d_1-l$,
since $|M_{H_{d_1-l+1}}^1\cap M_{H_{d_1}}^2|=f$,
it follows that there are at least $d_1-l$
copy of $H$, say $H_j$,
such that $|M\cap H_j|<d_2$,
and hence $M$ cannot monitor those edges of $H_j$,
a contradiction.
\end{proof}

Conversely, we suppose that $(1)$ and $(2)$ hold.
Since $(1)$ holds, it follows that $G$ has $k$
distance-edge-monitoring sets $M_G^i$,
where $k=\min\{s\,|\,\bigcup_{i=1}^s M_G^i=V(G)\}\geq 2$.
Since $(2)$ holds, it follows that $H$ has
$k$ distance-edge-monitoring sets $M_H^j$
such that for any two sets $M_H^p$ and $M_H^q$,
$M_H^p\cap M_H^q=\emptyset\ (1\leq i,j\leq k,\ i\neq j)$.
Without loss of generality, let
$M_H^j=\{v_j^1,v_j^2,\ldots,v_j^{d_2}\}$,
and $M_G^i=\{u_i^1,u_i^2,\ldots,u_i^{d_1}\}$.
Then $M=\cup_{i=1}^k (M_{G}^i\times M_{H}^i)$
is a distance-edge-monitoring set of $G\Box H$
and $|M|=nd_1$, and hence $\operatorname{dem}(G\Box H)\leq nd_1$.
It follows from Theorem \ref{GH} that
$\operatorname{dem}(G\Box H)\geq nd_1$,
and hence $\operatorname{dem}(G\Box H)=nd_1$.
\end{proof}

\subsection{Results for some specific networks}

In this section, we determine
distance-edge-monitoring numbers of some specific networks.
Foucaud et al. \cite{FF21} got a result for the Cartesian product of two paths.

\begin{theorem}{\upshape\cite{FF21}} \label{TPmPn}
For any integers $m, n\geq 2$, we have
$\operatorname{dem}(P_m \Box P_n) = \max\{m, n\}$.
\end{theorem}

Then the following is clear.

\begin{corollary}
Let $T_1$ and $T_2$ be two trees with order $m, n\geq 2$,
respectively. Then
$\operatorname{dem}(T_1 \Box T_2) = \max\{m, n\}$.
\end{corollary}

The Cartesian product of tree $T$ and cycle $C$ is
denoted by $T \Box C$.
Then we study $\operatorname{dem}(T \Box C)$.

\begin{theorem}\label{tc}
For tree $T$ and cycle $C$ with order
$m \geq 2$ and $n \geq 3$, respectively. Then
$$
\operatorname{dem}(T\Box C)=
\begin{cases}
 n & \mbox{if $n\geq 2m+1$,} \\
 2m & \mbox{if $n\leq 2m$.}
\end{cases}
$$
\end{theorem}
\begin{proof}
In $T\Box C$, there exist $m$ disjoint copies of $C$,
and $n$ disjoint copies of $T$. Note that each cycle $C$
needs two monitoring vertices, and each path $T$ needs
one monitoring vertex. Hence, for any two vertices
$w_{i,j}, w_{i,j'}$ of $T\Box C$,
where $v_j$ and $v_j'$ are not adjacent in $C_i$,
we have $EM(w_{i,j})\cup EM(w_{i,j'})
=E(T_j)\cup E(T_{j'})\cup E(C_i)$.

If $n\geq 2m+1$, then any distance-edge-monitoring
set $M$ contains not only two vertices of each copy
of $C$, but also a vertex of each copy of $T$.
Hence $|M|\geq n$, otherwise, there is an edge of
some copy of $T$ which cannot be monitored by $M$.
On the other hand, let
$M=\{w_{i,i},w_{i,i+m}\,|\,1\leq i\leq m\}
\cup\{w_{m,j}\,|\,2m+1\leq j\leq n\}$.
Obviously, any edge $e \in E(T \Box C)$ can
be monitored by $M$, so $\operatorname{dem}(T \Box C)\leq n$,
and hence $\operatorname{dem}(T \Box C)=n$.

If $n\leq 2m$, we have $|M|\geq 2m$, otherwise, there is an edge of
some copy of $C_n$ cannot be monitored by $M$.
Next we want to show
$\operatorname{dem}(T \Box C)\leq 2m$.
Let $$M=\{w_{i,i},w_{i,i+\lfloor n/2\rfloor}\,|\,1\leq i\leq \lfloor\frac{n}{2}\rfloor\}\cup\{w_{i,n},
w_{i,n-2}\,|\,\lfloor\frac{n}{2}\rfloor +1\leq i\leq m\}.$$
Obviously, any edge $e \in E(T\Box C)$ can
be monitored by $M$. Hence $\operatorname{dem}(T\Box C)=2m$.
\end{proof}

According to Theorem \ref{tc},
the following result is obvious.

\begin{corollary}\label{pc}
For path $P$ and cycle $C$ with order
$m \geq 2$ and $n \geq 3$, respectively.
Then
$$
\operatorname{dem}(P\Box C)=
\begin{cases}
 n & \mbox{if $n\geq 2m+1$,} \\
 2m & \mbox{if $n<2m+1$.}
\end{cases}
$$
\end{corollary}

\begin{theorem}
For a cycle $C$ and a path $P$
with order $m \geq 3$ and $n\geq 2$, respectively. If $\operatorname{dem}(C\Box P)=2n$,
then the following conditions hold.

$(1)$ For any vertex $x\in V(C)$,
there exists a distance-edge-monitoring set $M_{C}$
with $x\in M_{C}$.

$(2)$ $P$ has at least $k$ distance-edge-monitoring sets,
say $M_{P}^j\ (j=1,2,\ldots,k)$, such that
$M_{P}^a\cap M_{P}^b=\emptyset\ (1\leq a,b\leq k,\ a\neq b)$
for any two sets $M_{P}^a$ and $M_{P}^b$,
where $k=\min\{s\,|\,\bigcup_{i=1}^s M_{C}^i=V(C)\}\geq 2$.
\end{theorem}
\begin{proof}
It is clear, that for any two non-adjacent vertices $u_i,u_j\in V(C)$,
$\{u_i,u_j\}$ is a distance-edge-monitoring set of $C$.
So $(1)$ holds.
Then there exist some distance-edge-monitoring sets,
say $M_{C}^1,M_{C}^2,...,M_{C}^s$, such that
$\bigcup_{i=1}^s M_{C}^i=V(C)$.
If $i+2>m$, then let $u_{i+2}=u_j$,
where $i+2\equiv j\mod m$.
Let $M_{C}^i=\{u_i,u_{i+2}\}$.
Then $$k=\min\{s\,|\,\bigcup_{i=1}^s M_{C}^i=V(C)\}
=\left\lceil \frac{m}{2}\right\rceil .$$

It is clear that for any vertex
$v_j\in P$, $\{v_j\}$ is a
distance-edge-monitoring set of $P$.
Therefore, $P$ has $n$ distance-edge-monitoring sets,
and for any two distance-edge-monitoring sets
$M_{P}^{j_1}$ and $M_{P}^{j_2}$, we have
$M_{P}^{j_1}\cap P_{P}^{j_2}=\emptyset$.
It follows from Corollary \ref{pc}, that $m\leq 2n$.
So $k=\lceil \frac{m}{2}\rceil \leq n$,
and hence $(2)$ holds.
\end{proof}

Next we study the Cartesian product between
the cycles.

\begin{theorem}\label{cc}
For two cycles $C^1$ and $C^2$ with
order $m,n \geq 3$, respectively. Then
$$
dem(C^1\Box C^2)=\max\{2m,2n\}.
$$
\end{theorem}
\begin{proof}
Without loss of generality, suppose that $n\geq m$.
It suffices to prove that $dem(C^1\Box C^2)=2n$.
Suppose that $M'$ is a distance-edge-monitoring set of
$C^1\Box C^2$ with $|M'|=2n-1$.
Then there exists a subgraph $C^1_a$ such that
$|M'\cap V(C^1_a)|=1$,
and so there is at least an edge $e\in E(C^1_a)$,
which cannot be monitored by $M'$, a contradiction.
Hence $dem(C^1\Box C^2)\geq 2n$.

Next, we need to show that $dem(C^1\Box C^2)\leq 2n$.
If $i+2>m$, then let $u_{i+2}=u_j$,
where $i+2\equiv j\mod m$.
Let $M_i=\{w_{i,i},w_{i+2,i}\}$,
then $\bigcup_{i=1}^{m} M_i$ can monitor
all the edges of
$\bigcup_{i=1}^m\left(E(C^1_i)\cup E(C^2_i)\right)$.
So let $$M=\left(\bigcup_{i=1}^{m} M_i\right)\cup
\left(\bigcup_{j=m+1}^{n}\{w_{1,j},w_{3,j}\}\right),$$
then for any edge $e \in E(C^1\Box C^2)$ can
be monitored by $M$, and $|M|=2n$.
So $dem(C^1\Box C^2)\leq 2n$,
and hence $dem(C^1\Box C^2)=2n$.
The case $n\leq m$ can be also proved similarly.
\end{proof}

Finally, we determine the Cartesian product between
two complete graphs.

\begin{theorem}\label{kk}
For two complete graphs $K^1$ and $K^2$ with
order $m,n \geq 3$, respectively.
Then
$$
\operatorname{dem}(K^1\Box K^2)=mn-\min\{m,n\}.
$$
\end{theorem}
\begin{proof}
Denote the vertex set of $K^1\Box K^2$ by
$V(K^1\Box K^2)=\{w_{i,j}|u_i\in V(K^1), v_j\in V(K^2)\}$.
Suppose that $m\geq n$. It suffices to show that
$\operatorname{dem}(K^1\Box K^2)=mn-n$.
Let $M_i$ be the distance-edge-monitoring set
of $K^1_i\ (i=1,2,\ldots,n)$, and
$M_i=V(K^1_i)-\{w_{i,j}\}$.
It is clear that $M=\bigcup_{i=1}^n M_i$
is a distance-edge-monitoring set of
$\operatorname{dem}(K^1\Box K^2)$.
Moreover, $|M|=mn-n$, and hence
$\operatorname{dem}(K^1\Box K^2)\leq mn-n$.
Next we need to show that $\operatorname{dem}(K^1\Box K^2)\geq mn-n$.
Since the vertices of each $K^1_i$ cannot
monitor the edges of $K^1_j\ (i=1,2,\ldots,n$
and $j\neq i)$, and each $K^1_i$ needs at least
$n-1$ vertices to monitor all the edges of $K^1_i$,
it follows that $n$ complete graphs $K^1_i$ need
at least $mn-n$ vertices to monitor all the edges.
We have that $\operatorname{dem}(K^1\Box K^2)\geq mn-n$,
and hence $\operatorname{dem}(K^1\Box K^2)=mn-n$.
The case $m\leq n$ can be also proved similarly.
\end{proof}

\subsection{Comparison results}

In this section, we will compare the distance-edge-monitoring number
with other parameters which are also related to distance.
First, we give the definitions of these parameters in the following.

Let $G=(V(G), E(G))$ be a simple connected graph.
A set $S\subseteq V(G)$ is called a metric generator
of $G$ if for any two distinct vertices $u,v\in V(G)$,
there is a vertex $s\in S$ such that $d_G(u, s)\neq d_G(v, s)$.
The minimum cardinality of a metric generator is called the metric dimension of $G$ and denoted by $\operatorname{dim}(G)$.

The distance between the vertex $v\in V(G)$
and the edge $e=uw\in E(G)$ is defined
as $d_G(e,v)=\min\{d_G(u,v), d_G(w,v)\}$.
The vertex $v\in V(G)$ distinguishes two edges $e_1, e_2\in E(G)$ if
$d_G(w, e_1)\neq d_G(w, e_2)$.
A nonempty set $S\subset V$ is an edge metric generator for $G$ if
any two edges of $G$ are distinguished by some vertex of $S$.
The edge metric dimension of $G$,
denoted by $\operatorname{edim}(G)$,
is the smallest cardinality of an edge metric generator.

For two vertices $u,v\in V(G)$,
the interval $I_G[u, v]$ between $u$ and $v$
is defined as the collection of all vertices that
belong to some shortest $u-v$ path.
A vertex $w$ strongly resolves two vertices $u$ and $v$
if $v\in I_G[u, w]$ or $u\in I_G[v, w]$.
A vertex set $S$ of $G$ is
a strong resolving set if any two vertices of $G$
are strongly resolved by some vertex of $S$.
The smallest cardinality of a strong resolving set of $G$
is called strong metric dimension and denoted by
$\operatorname{dim_s}(G)$.

In Table \ref{T-1} we compare the distance-edge-monitoring number
with (edge) metric dimension and strong metric dimension
in some networks.
In Table \ref{T-2} we compare the distance-edge-monitoring number
with edge metric dimension in some operations of two graphs.
In Fig. \ref{grid} we compare the distance-edge-monitoring number
with the (edge) metric dimension on $P_m\Box P_n$.
In Fig. \ref{complete} we compare the distance-edge-monitoring number
with the metric dimension and the strong metric dimension in $K_m \Box K_n$.

\begin{table*}[!htbp]
\centering
\caption{The comparison in some networks}\label{T-1}
\resizebox{\textwidth}{!}{
\begin{tabular}{|c|c|c|c|c|}
\hline
Parameter  & Grid graph $P_m \Box P_n$  & Torus $C_m \Box C_n$ &
$K_m \Box K_n$ & Hypercube $Q_n$  \\\hline
The DEM number & $\max\{m, n\}$ &
$\max\{2m, 2n\}$  & $mn-\min\{m,n\}$  & $2^{n-1}$ \\\hline
The metric dimension & $2$\cite{KTY}  & $4$ or $3$\cite{JCMIMCD}  &
$\lfloor\frac{2}{3}(m+n-1)\rfloor$ or $m-1$\cite{JCMIMCD}  & * \\\hline
The edge metric dimension & $2$\cite{KTY} &
$3$, where $m=4t,n=4r$\cite{KTY} & * & *  \\\hline
The strong metric dimension &  * &
* & $\min\{m(n-1),n(m-1)\}$\cite{RY} & *  \\\hline
\end{tabular}
}
\end{table*}

\begin{table*}[!htbp]
\centering
\caption{The comparison in some operations of two graphs}\label{T-2}
\resizebox{\textwidth}{!}{
\begin{tabular}{|c|c|c|c|}
\hline
Operation  & The DEM number &
The edge metric dimension \\\hline
Join $G \vee H$  &
$\min\{c(G)+|V(H)|,\,c(H)+|V(G)|\}$ &
$|V(G)|+|V(H)|-1$ or $|V(G)|+|V(H)|-2$\cite{II} \\\hline
Corona $G * H$ &
$|V(G)|\operatorname{dem}(H)$ &
$|V(G)|\cdot(|V(H)|-1)$\cite{II} \\\hline
Cluster $G \odot H$ &
$\operatorname{dem}(G)\leq \operatorname{dem}(G\odot H)
\leq |V(G)|\operatorname{dem}(H)$ &
*  \\\hline
$G \Box H$ &
\makecell{$\max\{m\operatorname{dem}(H),n\operatorname{dem}(G)
\leq\operatorname{dem}(G\Box H)\}$\\
$\leq m\operatorname{dem}(H)+n\operatorname{dem}(G)
-\operatorname{dem}(G)\operatorname{dem}(H)$} &
*  \\\hline
\end{tabular}
}
\end{table*}

In above two tables, $*$ means that the problem is still open.
All the results in above two tables show that
these parameters which related to distance are different.
Therefore, research on this issue is of great significance.
Moreover, one can study the relationship between these parameters.

\begin{figure}
  \centering
  \includegraphics[width=0.6\textwidth,height=60mm]{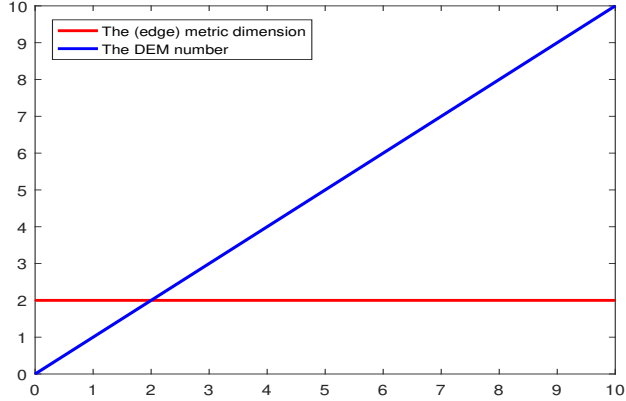}
  \caption{The comparison of different parameter on $P_m\Box P_n$}\label{grid}
\end{figure}

\begin{figure}[!htbp]
\centering
\begin{minipage}{0.5\linewidth}
\vspace{3pt}
\centerline{
\includegraphics[width=\textwidth]{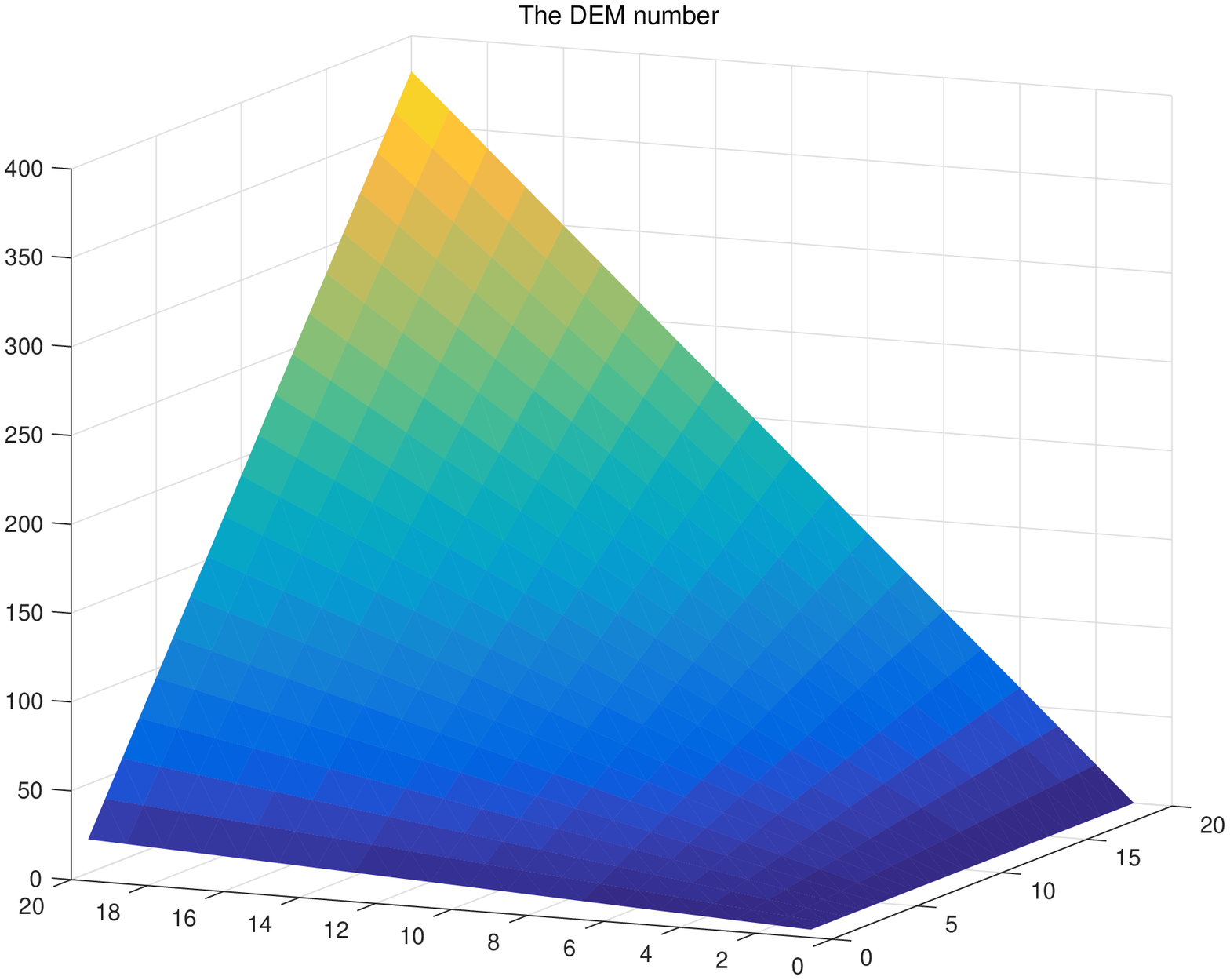}
}
\end{minipage}
\begin{minipage}{0.49\linewidth}
\vspace{3pt}
\centerline{
\includegraphics[width=\textwidth]{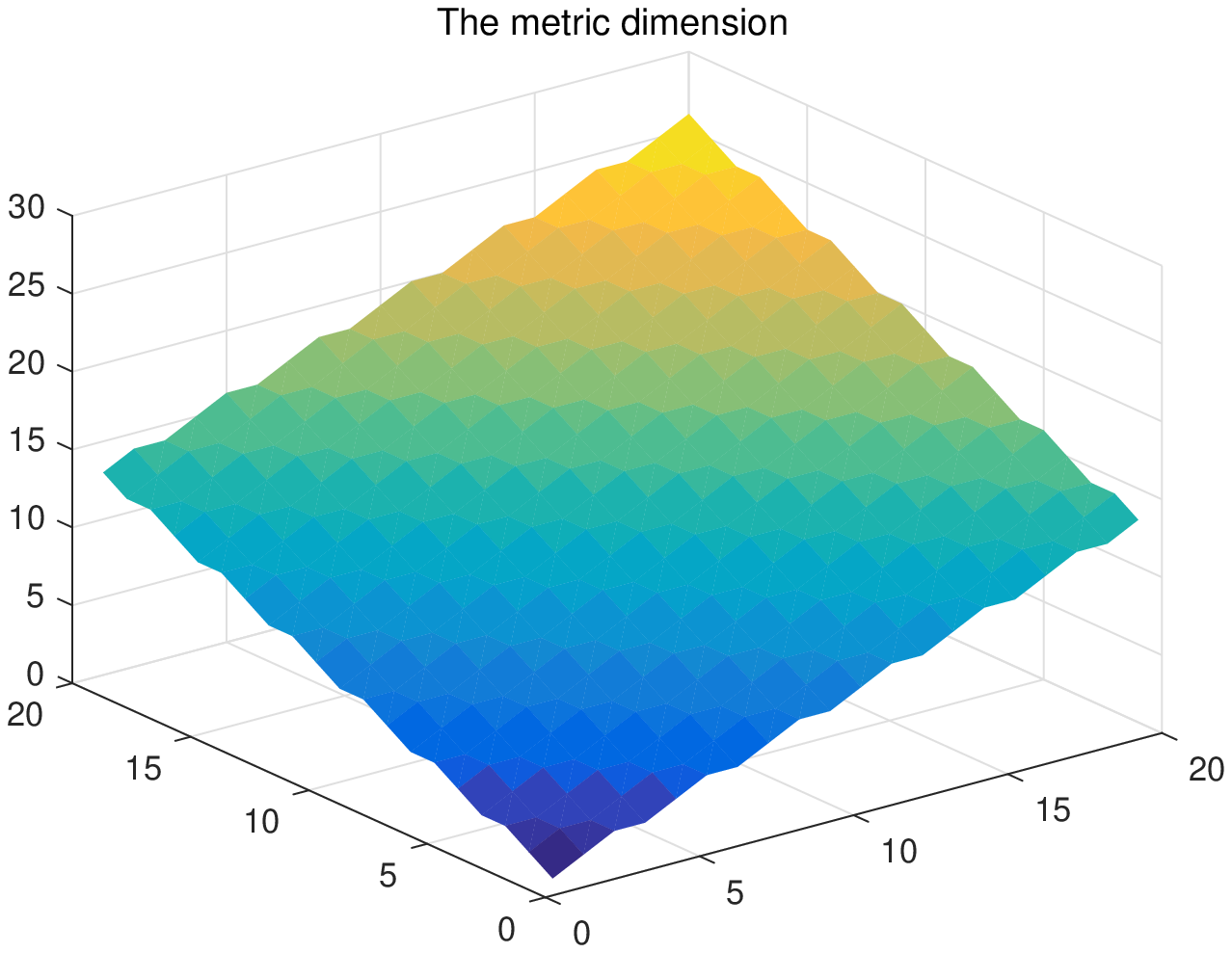}
}
\end{minipage}
\begin{minipage}{0.5\linewidth}
\vspace{3pt}
\centerline{
\includegraphics[width=\textwidth]{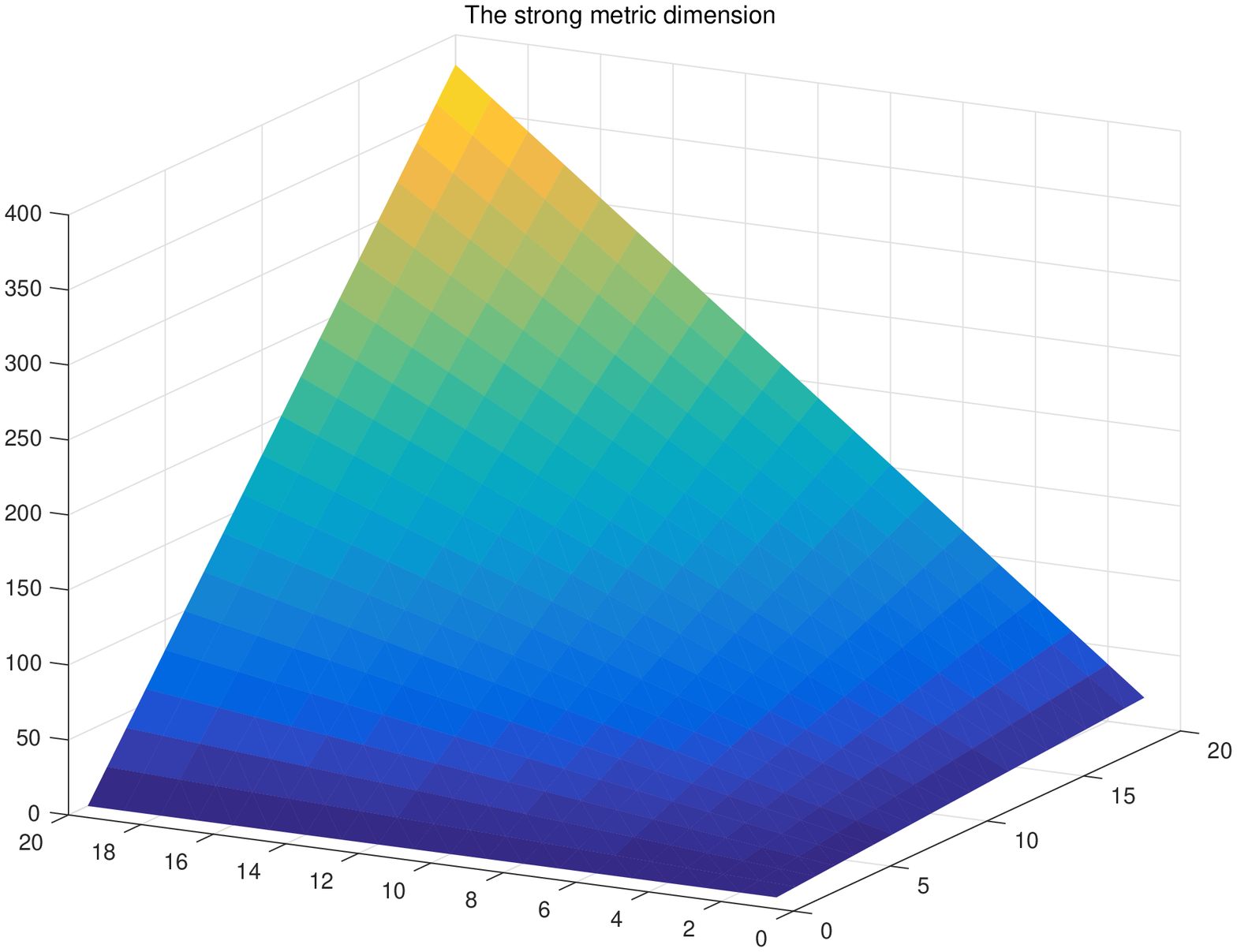}
}
\end{minipage}

\caption{The comparison of different parameter on
$K_m \Box K_n(m\leq n\leq 2m-1)$}\label{complete}

\end{figure}

\section{Conclusion}

In this paper, we have continued the study of {\it distance-edge-monitoring sets},
a new graph parameter recently introduced by Foucaud {\it et al.}~\cite{FF21},
which is useful in the area of network monitoring. In particular, we have
studied the distance-edge-monitoring numbers of join, corona, cluster,
and Cartesian products, and obtained the exact values of some specific networks.

For future work, it would be interesting to study distance-edge monitoring
sets in further standard graph classes, including pyramids, Sierpi\'nki-type graphs, circulant graphs,
or line graphs. In addition, characterizing the graphs with $dem(G)=n-2$
would be of interest, as well as clarifying further the relation of the
parameter $dem(G)$ to other standard graph parameters, such as
arboricity, vertex cover number and feedback edge set number.

\end{document}